\documentclass{llncs}
\usepackage{amsmath,amssymb}
\usepackage[dvipdfm]{graphicx}

\newtheorem{observation}{Observation}
\newcommand{\msize}[1]{{\left|#1\right|}}
\newcommand{\PP}{{\sf P}}
\newcommand{\MPQ}{{\sf MPQ}}
\newcommand{\calI}{{\cal I}}
\newcommand{\calT}{{\cal T}}
\newcommand{\bfI}{{\mathbf I}}

\newcommand{\fmap}[2]{f_{#1}(#2)}
\newcommand{\fallmap}{g}

\newcommand{\LCA}{{\sf LCA}}
\newcommand{\calS}{{\cal S}}

\newenvironment{listing}[1]{%
        \begin{list}{*}{%
                 \settowidth{\labelwidth}{#1}%
                 \setlength{\leftmargin}{\labelwidth}%
                 \advance \leftmargin by 12pt
                   \setlength{\itemsep}{0pt}%
                   \setlength{\parsep}{0pt}%
                   \setlength{\topsep}{0pt}%
                   \setlength{\parskip}{0pt}%
}%
}{%
\end{list}}

\begin{document}
\title{Shortest Reconfiguration of Sliding Tokens on a Caterpillar}

\author{Takeshi Yamada\inst{1} \and Ryuhei Uehara\inst{1}}

\institute{School of Information Science, JAIST, Japan.\\
    \email{\{tyama,uehara\}@jaist.ac.jp}}

\maketitle

\begin{abstract}
Suppose that we are given two independent sets $\bfI_b$ and $\bfI_r$ of 
a graph such that $\msize{\bfI_b}=\msize{\bfI_r}$, 
and imagine that a token is placed on each vertex in $\bfI_b$. 
Then, the {\sc sliding token} problem is to determine whether 
there exists a sequence of independent sets which transforms $\bfI_b$ 
into $\bfI_r$ so that each independent set in the sequence results from 
the previous one by sliding exactly one token along an edge in the graph. 
The {\sc sliding token} problem is one of the reconfiguration problems
that attract the attention from the viewpoint of theoretical computer science.
The reconfiguration problems tend to be PSPACE-complete in general,
and some polynomial time algorithms are shown in restricted cases.
Recently, the problems that aim at finding a shortest reconfiguration sequence are investigated.
For the 3SAT problem, a trichotomy for the complexity of finding the shortest sequence has been shown; 
that is, it is in P, NP-complete, or PSPACE-complete in certain conditions.
In general, even if it is polynomial time solvable to decide whether two instances are reconfigured with each other, 
it can be NP-complete to find a shortest sequence between them.
Namely, finding a shortest sequence between two independent sets 
can be more difficult than the decision problem of reconfigurability between them.
In this paper, we show that the problem for finding a shortest sequence between two independent sets 
is polynomial time solvable for some graph classes which are subclasses of the class of interval graphs.
More precisely, we can find a shortest sequence between two independent sets on a graph $G$ 
in polynomial time if either $G$ is a proper interval graph, a trivially perfect graph, or a caterpillar.
As far as the authors know, this is the first polynomial time algorithm for the {\sc shortest sliding token}
 problem for a graph class that requires detours.
%
\end{abstract}

\section{Introduction}

Recently, the {\em reconfiguration problems} attract the attention 
from the viewpoint of theoretical computer science.
The problem arises when we wish to find a step-by-step transformation between 
two feasible solutions of a problem such that 
all intermediate results are also feasible and 
each step abides by a fixed reconfiguration rule, that is, 
an adjacency relation defined on feasible solutions of the original problem.
The reconfiguration problems have been studied extensively for several well-known problems, including 
{\sc independent set}~\cite{HearnDemaine2005,HearnDemaine2009,IDHPSUU,KaminskiMedvedevMilanic2012,MNRSS13}, 
{\sc satisfiability}~\cite{Kolaitis,MTY11}, 
{\sc set cover}, {\sc clique}, {\sc matching}~\cite{IDHPSUU}, 
{\sc vertex-coloring}~\cite{BJLPP11,BC09,CHJ11}, {\sc shortest path}~\cite{KMP11}, and so on.

\begin{figure}[t]
\begin{center}
\includegraphics[width=0.8\textwidth]{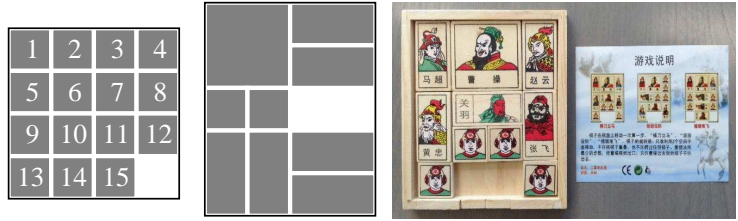}
\end{center}
\caption{The 15 puzzle, Dad's puzzle, and its Chinese variant.}
\label{fig:dad}
\end{figure}

The reconfiguration problem can be seen as a natural ``puzzle'' from the viewpoint of 
recreational mathematics. The {\em 15 puzzle} is one of the most famous classic puzzles,
that had the greatest impact on American and European society of any mechanical puzzle 
the word has ever known in 1880 (see \cite{Slocum} for its rich history).
It is well known that the 15 puzzle has a parity; 
for any two placements, we can decide whether two placements are reconfigurable or not by checking the parity.
Therefore, we can solve the reconfiguration problem in linear time just by 
checking whether the parity of one placement coincides with the other or not.
Moreover, we can say that the distance between any two reconfigurable placements is $O(n^3)$, 
that is, we can reconfigure from one to the other in $O(n^3)$ sliding pieces when the size of the board is $n\times n$.
However, surprisingly, for these two reconfigurable placements, 
finding a shortest path is NP-complete in general \cite{RW90}.
Namely, although we know that it is $O(n^3)$, finding a shortest one is NP-complete.
Another interesting property of the 15 puzzle is in another case of generalization.
In the 15 puzzle, every peace has the same unit size of $1\times 1$.
We have the other famous classic puzzles that can be seen as a generalization of this viewpoint.
That is, when we allow to have rectangles, we have the other classic puzzles, 
called ``Dad puzzle'' and its variants (see \figurename~\ref{fig:dad}).
Gardner said that ``These puzzles are very much in want of a theory'' in 1964 \cite{Gardner},
and Hearn and Demaine have gave the theory after 40 years \cite{HearnDemaine2005};
they prove that these puzzles are PSPACE-complete in general 
using their nondeterministic constraint logic model \cite{HearnDemaine2009}.
That is, the reconfiguration of the sliding block puzzle is PSPACE-complete in general decision problem,
and linear time solvable if every block is the unit square. 
However, finding a shortest reconfiguration for the latter easy case is NP-complete.
In other words, we can characterize these three complexity classes using the model of sliding block puzzle.

From the viewpoint of theoretical computer science, one of the most important problems is the 3SAT problem.
For this 3SAT problem, a similar trichotomy for the complexity of finding a shortest sequence has been shown
recently; that is, for the reconfiguration problem of 3SAT, finding a shortest sequence between two satisfiable
assignments is in P, NP-complete, or PSPACE-complete in certain conditions \cite{MNPR15}.
In general, the reconfiguration problems tend to be PSPACE-complete,
and some polynomial time algorithms are shown in restricted cases.
In the reconfiguration problems, finding a shortest sequence can be a new trend in theoretical computer
science because it has a great potential to characterize the class NP from a new viewpoint.

Beside the 3SAT problem, one of the most important problems in theoretical computer science is the
independent set problem. 
Recall that an {\em independent set} of a graph $G$ is a vertex-subset of $G$ in which 
no two vertices are adjacent. 
(See \figurename~\ref{fig:example} which depicts five different independent sets in the same graph.)
For this notion, the natural reconfiguration problem is called the {\sc sliding token} problem introduced by 
Hearn and Demaine~\cite{HearnDemaine2005}:
Suppose that we are given two independent sets $\bfI_b$ and $\bfI_r$ of a graph 
$G = (V,E)$ such that $\msize{\bfI_b}=\msize{\bfI_r}$, 
and imagine that a token (coin) is placed on each vertex in $\bfI_b$. 
Then, the {\sc sliding token} problem is to determine whether there exists a sequence 
$\langle \bfI_1, \bfI_2, \ldots, \bfI_{\ell} \rangle$ of independent sets of $G$ such that
\begin{listing}{aaa}
\item[(a)] $\bfI_1=\bfI_b$, $\bfI_{\ell}=\bfI_r$, 
 and $\msize{\bfI_i} = \msize{\bfI_b}=\msize{\bfI_r}$ for all $i$, $1 \le i \le \ell$; and 
\item[(b)] for each $i$, $2 \le i \le \ell$, 
 there is an edge $\{u,v\}$ in $G$ such that $\bfI_{i-1} \setminus\bfI_{i}=\{u\}$ 
 and $\bfI_{i}\setminus\bfI_{i-1}=\{v\}$, 
 that is, $\bfI_{i}$ can be obtained from $\bfI_{i-1}$ by sliding exactly 
 one token on a vertex $u \in \bfI_{i-1}$ to its adjacent vertex $v$ along $\{u,v\} \in E$.
\end{listing}
Figure~\ref{fig:example} illustrates a sequence 
$\langle \bfI_1, \bfI_2, \ldots, \bfI_5 \rangle$ of independent sets 
which transforms $\bfI_b = \bfI_1$ into $\bfI_r = \bfI_5$. 
Hearn and Demaine proved that the {\sc sliding token} problem is PSPACE-complete for planar graphs.

\begin{figure}[t]
\begin{center}
\includegraphics[width=0.8\textwidth]{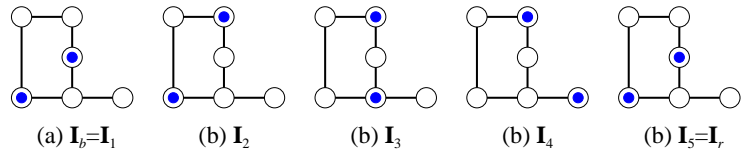}
\end{center}
\caption{A sequence $\langle \bfI_1, \bfI_2, \ldots, \bfI_5 \rangle$ of independent sets of the same graph,
 where the vertices in independent sets are depicted by small black circles (tokens).}
\label{fig:example}
\end{figure}

(We note that the reconfiguration problem for {\sc independent set} have some variants.
In~\cite{KaminskiMedvedevMilanic2012},
the reconfiguration problem for {\sc independent set} is studied
under three reconfiguration rules called
``token sliding,'' ``token jumping,'' and ``token addition and removal.''
In this paper, we only consider token sliding model, and 
see \cite{KaminskiMedvedevMilanic2012} for the other models.)

For the {\sc sliding token} problem, some polynomial time algorithms are 
investigated as follows: Linear time algorithms have been shown for
cographs (also known as $P_4$-free graphs) \cite{KaminskiMedvedevMilanic2012} 
and trees \cite{DDFEHIOOUY2015}. 
Polynomial time algorithms are shown 
for bipartite permutation graphs \cite{FoxEpsteinHoangOtachiUehara2015}, 
and claw-free graphs \cite{BonsmaKaminskiWrochna}.
On the other hand, PSPACE-completeness is also shown for 
graphs of bounded tree-width \cite{MouawadNishimuraRamanWrochna}, and planar graphs \cite{HearnDemaine2009}.

In this context, we investigate for finding a shortest sequence of the {\sc sliding token} problem,
which is called the {\sc shortest sliding token} problem. That is, our problem is formalized as follows:
\begin{listing}{aaa}
\item[Input:] A graph $G=(V,E)$ and two independent sets $\bfI_b,\bfI_r$ with $\msize{\bfI_b}=\msize{\bfI_r}$.
\item[Output:] A {\em shortest} reconfiguration sequence $\bfI_b=\bfI_1$, $\bfI_2$, $\ldots$, $\bfI_{\ell}=\bfI_r$ such that
	     $\bfI_{i}$ can be obtained from $\bfI_{i-1}$ by sliding exactly 
	     one token on a vertex $u \in \bfI_{i-1}$ to its adjacent vertex $v$ along $\{u,v\} \in E$ for
	     each $i$, $2 \le i \le \ell$.
\end{listing}
We note that $\ell$ is not necessarily in polynomial of $\msize{V}$;
this is an issue how we formalize the problem, and if we do not know that $\ell$ is in polynomial or not.
If the length $k$ is given as a part of input, we may be able to decide whether $\ell \le k$ in polynomial time 
even if $\ell$ itself is not in polynomial. However, if we have to output the sequence itself, 
it cannot be solved in polynomial time if $\ell$ is not in polynomial.


In this paper, we will show that the {\sc shortest sliding token} problem is solvable in polynomial time
for the following graph classes:
\paragraph{Proper interval graphs:}
We first prove that any two independent sets of 
a proper interval graph can be transformed into each other. 
In other words, every proper interval graph with two independent sets 
$\bfI_b$ and $\bfI_r$ is a yes-instance of the problem if $\msize{\bfI_b}=\msize{\bfI_r}$.
Furthermore, we can find the ordering of tokens to be slid in a minimum-length sequence in $O(n)$ time (implicitly),  
even though there exists an infinite family of independent sets on paths 
(and hence on proper interval graphs) for which any sequence requires $\Omega(n^2)$ length.

\paragraph{Trivially perfect graphs:}
We then give an $O(n)$-time algorithm for trivially perfect graphs 
which actually finds a shortest sequence if such a sequence exists. 
In contrast to proper interval graphs, 
any shortest sequence is of length $O(n)$ for trivially perfect graphs. 
Note that trivially perfect graphs form a subclass of cographs, 
and hence its polynomial time solvability has been known~\cite{KaminskiMedvedevMilanic2012}.

\paragraph{Caterpillars:}
We finally give an $O(n^2)$-time algorithm for caterpillars for the shortest sliding token problem.
To make self-contained, we first show a linear time algorithm for decision problem that
asks whether two independent sets can be transformed into each other.
(We note that this problem can be solved in linear time for a tree \cite{DDFEHIOOUY2015}.)
For a yes-instance, we next show an algorithm that finds 
a shortest sequence of token sliding between two independent sets.

We here remark that, since the problem is PSPACE-complete in general,
an instance of the {\sc sliding token} problem may require 
the exponential number of independent sets to transform. 
In such a case, tokens should make detours to avoid violating to be independent (as shown in \figurename~\ref{fig:example}).
As we will see, caterpillars certainly require to make detours to transform.
Therefore, it is remarkable that any yes-instance on 
a caterpillar requires a sequence of token-slides of polynomial length.
This is still open even for a tree. 
That is, in a tree, we can determine if two independent sets are reconfigurable in linear time due to 
\cite{DDFEHIOOUY2015}, however, we do not know if the length of the sequence is in polynomial.

As far as the authors know, this is the first polynomial time algorithm for the {\sc shortest sliding token}
 problem for a graph class that requires detours of tokens.

\section{Preliminaries}

In this section, we introduce some basic terms and notations. 
In the {\sc sliding token} problem, 
we may assume without loss of generality that graphs are simple and connected.
For a graph $G=(V,E)$, we let $n=\msize{V}$ and $m=\msize{E}$.

\subsection{{\sc Sliding token}}

For two independent sets $\bfI_i$ and $\bfI_j$ of the same cardinality in a graph $G=(V,E)$, 
if there exists exactly one edge $\{u,v\}$ in $G$ such that $\bfI_{i} \setminus\bfI_{j}=\{u\}$ 
and $\bfI_{j}\setminus\bfI_{i}=\{v\}$, 
then we say that $\bfI_{j}$ can be obtained from $\bfI_{i}$ by {\em sliding} a token on 
the vertex $u \in \bfI_{i}$ to its adjacent vertex $v$ along the edge $\{u,v\}$, 
and denote it by $\bfI_{i} \vdash \bfI_{j}$. 
We remark that the tokens are unlabeled, while the vertices in a graph are labeled.

A {\em reconfiguration sequence} between two independent sets $\bfI_1$ and $\bfI_{\ell}$ of $G$ 
is a sequence $\langle \bfI_1, \bfI_2, \ldots, \bfI_{\ell} \rangle$ of 
independent sets of $G$ such that $\bfI_{i-1} \vdash \bfI_i$ for $i=2, 3, \ldots, \ell$.
We denote by $\bfI_{1} \vdash^* \bfI_{\ell}$ if there exists 
a reconfiguration sequence between $\bfI_1$ and $\bfI_{\ell}$.
We note that a reconfiguration sequence is {\em reversible}, that is,
we have $\bfI_{1} \vdash^* \bfI_{\ell}$ if and only if $\bfI_{\ell} \vdash^* \bfI_{1}$.
Thus we say that two independent sets $\bfI_1$ and $\bfI_{\ell}$ are 
{\em reconfigurable} into each other if $\bfI_{1} \vdash^* \bfI_{\ell}$.
The {\em length} of a reconfiguration sequence $\calS$ is defined as 
the number of independent sets contained in $\calS$.
For example, the length of the reconfiguration sequence in \figurename~\ref{fig:example} is $5$. 

The {\sc sliding token} problem is to determine whether two given independent sets 
$\bfI_b$ and $\bfI_r$ of a graph $G$ are reconfigurable into each other. 
We may assume without loss of generality that $\msize{\bfI_b} = \msize{\bfI_r}$; 
otherwise the answer is clearly ``no.''
Note that the {\sc sliding token} problem is a decision problem asking for the existence of a reconfiguration sequence 
between $\bfI_b$ and $\bfI_r$, and hence it does not ask an actual reconfiguration sequence. 
In this paper, we will consider the {\sc shortest sliding token} problem that 
computes a shortest reconfiguration sequence between two independent sets.
Note that the length of a reconfiguration sequence may not be in polynomial of the size of the graph
since the sequence may contain detours of tokens.

We always denote by $\bfI_b$ and $\bfI_r$ the initial and target independent sets of $G$, 
respectively, as an instance of the {\sc (shortest) sliding token} problem;
we wish to slide tokens on the vertices in $\bfI_b$ to the vertices in $\bfI_r$. 
We sometimes call the vertices in $\bfI_b$ {\em blue}, 
and the vertices in $\bfI_r$ {\em red}; 
each vertex in $\bfI_b\cap\bfI_r$ is blue {\em and} red.

\subsection{Target-assignment}
	
We here give another notation of the {\sc sliding token} problem, 
which is useful to explain our algorithm.
	
Let $\bfI_b=\{b_1,b_2,\ldots,b_k\}$ be an initial independent set of a graph $G$.
For the sake of convenience, we label the tokens on the vertices in $\bfI_b$; 
let $t_i$ be the token placed on $b_i$ for each $i$, $1 \le i \le k$.
Let $\calS$ be a reconfiguration sequence between $\bfI_b$ and 
an independent set $\bfI$ of $G$, and hence $\bfI_b \vdash^* \bfI$.
Then, for each token $t_i$, $1 \le i \le k$, 
we denote by $\fmap{\calS}{t_i}$ the vertex in $\bfI$ on which 
 the token $t_i$ is placed via the reconfiguration sequence $\calS$. 
 Notice that $\{ \fmap{\calS}{t_i} \mid 1 \le i \le k\} = \bfI$. 
	
Let $\bfI_r$ be a target independent set of $G$, 
 which is not necessarily reconfigurable from $\bfI_b$. 
Then, we call a mapping $\fallmap: \bfI_b \to \bfI_r$ 
 a {\em target-assignment} between $\bfI_b$ and $\bfI_r$.
The target-assignment $\fallmap$ is said to be {\em proper} 
 if there exists a reconfiguration sequence $\calS$ such that 
 $\fmap{\calS}{t_i} = \fallmap(b_i)$ for all $i$, $1 \le i \le k$. 
Note that there is no proper target-assignment between $\bfI_b$ and $\bfI_r$ 
 if $\bfI_b \not\vdash^* \bfI_r$. 
Therefore, the {\sc sliding token} problem can be seen as the problem of 
 determining whether there exists at least 
 one proper target-assignment between $\bfI_b$ and $\bfI_r$. 

\subsection{Interval graphs and subclasses}

The {\em neighborhood} of a vertex $v$ in a graph $G=(V,E)$ is 
the set of all vertices adjacent to $v$, 
and we denote it by $N(v) = \{u\in V \mid \{u,v\}\in E\}$.
Let $N[v] = N(v)\cup\{v\}$.
For any graph $G=(V,E)$, two vertices $u$ and $v$ are called {\em strong twins} if $N[u]=N[v]$,
and {\em weak twins} if $N(u)=N(v)$.
In our problem, strong twins have no meaning: 
when $u$ and $v$ are strong twins, only one of them can be used by a token.
Therefore, in this paper, we only consider the graphs without strong twins.
That is, for any pair of vertices $u$ and $v$, we have $N[u]\neq N[v]$. 
(We have to take care about weak twins; see Section \ref{sec:caterpillars} for the details.)




A graph $G = (V,E)$ with $V = \{v_1,v_2,\ldots,v_n\}$ is an {\em interval graph} 
if there exists a set $\calI$ of (closed) intervals $I_1,I_2,\ldots,I_n$ 
such that $\{v_i,v_j\}\in E$ if and only if $I_i\cap I_j\neq\emptyset$ 
for each $i$ and $j$ with $1\le i,j\le n$.\footnote{In this paper, 
a bold $\bfI$ denotes an ``independent set,'' an italic $I$ denotes an 
``interval,'' and calligraphy $\calI$ denotes ``a set of intervals.''}
We call the set $\calI$ of intervals an {\em interval representation} of the graph, 
and sometimes identify a vertex $v_i \in V$ with its corresponding interval $I_i \in \calI$.
We denote by $L(I)$ and $R(I)$ the left and right endpoints of 
an interval $I \in \calI$, respectively.
That is, we always have $L(I)\le R(I)$ for any interval $I=[L(I),R(I)]$.

To specify the bottleneck of the running time of our algorithms, 
we suppose that an interval graph $G=(V,E)$ is given as an input by 
its interval representation using $O(n)$ space.
(If necessary, an interval representation of $G$ can be found in $O(n+m)$ time~\cite{KM89}.) 
More precisely, $G$ is given by a string of length $2n$ over alphabets 
$\{L(I_1), L(I_2), \ldots, L(I_n), R(I_1), R(I_2), \ldots, \allowbreak R(I_n)\}$.
For example, a complete graph $K_3$ with three vertices can be given by 
an interval representation $L(I_1) L(I_2) L(I_3) \allowbreak R(I_1) R(I_2) R(I_3)$,
and a path of length two is given by an interval representation $L(I_1) L(I_2) R(I_1) L(I_3) R(I_2) R(I_3)$. 

	
An interval graph is {\em proper} if it has an interval representation 
such that no interval properly contains another. 
The class of proper interval graphs is also known as 
the class of unit interval graphs~\cite{BogartWest1999}:
an interval graph is {\em unit} if it has an interval representation 
such that every interval has unit length. 
Hereafter, we assume that each proper interval graph is given in 
the interval representation of intervals of unit length.
In the context of the interval representation, 
an interval graph is proper if and only if $L(I_i)<L(I_j)$ if and only if $R(I_i)<R(I_j)$.

An interval graph is {\em trivially perfect} if 
it has an interval representation such that the relationship between 
any two intervals is either disjoint or inclusion.
That is, for any two intervals  $I_i$ and $I_j$ with $L(I_i)<L(I_j)$, we have either
$L(I_i)<L_(I_j)<R(I_j)<R(I_i)$ or $L(I_i)<R(I_i)<L(I_j)<R(I_j)$.

A {\em caterpillar} $G=(V,E)$ is a tree (i.e., a connected acyclic graph) 
that consists of two subsets $S$ and $L$ of $V$ as follows.
The vertex set $S$ induces a path $(s_1,\ldots,s_{n'})$ in $G$,
and each vertex $v$ in $L$ has degree 1, and its unique neighbor is in $S$.
We call the path $(s_1,\ldots,s_{n'})$ {\em spine}, and each vertex in $L$ {\em leaf}.
In this paper, without loss of generality, we assume that $n'\ge 2$,
$\deg(s_1)\ge 2$, and $\deg(s_{n'})\ge 2$. 
That is, the endpoints $s_1$ and $s_{n'}$ of 
the spine $(s_1,\ldots,s_{n'})$ should have at least one leaf.
It is easy to see that the class of caterpillars is a proper subset of the class of interval graphs,
and these three subclasses are incomparable with each other.

\section{Proper Interval Graphs}
\label{sec:proper}

We show the main theorem in this section for proper interval graphs, 
which first says that the answer of {\sc sliding token} is always ``yes'' for connected proper interval graphs.
We give a constructive proof of the claim, and it certainly finds a shortest sequence in linear time.
\begin{theorem}
\label{th:proper}
For a connected proper interval graph $G=(V,E)$,
any two independent sets $\bfI_{b}$ and $\bfI_{r}$ with $\msize{\bfI_b} = \msize{\bfI_r}$ 
are reconfigurable into each other, that is, $\bfI_{b}\vdash^* \bfI_{r}$.
Moreover, the shortest reconfiguration sequence can be found in polynomial time.
\end{theorem}

We give a constructive proof for Theorem~\ref{th:proper}, that is, 
we give an algorithm which actually finds a shortest reconfiguration sequence between 
any two independent sets $\bfI_b$ and $\bfI_r$ of a connected proper interval graph $G$.

A connected proper interval graph $G=(V,E)$ has a unique interval representation (up to reversal), and
we can assume that each interval is of unit length in the representation \cite{DengHellHuang1996}.
Therefore, by renumbering the vertices, we can fix an interval representation 
$\calI =\{I_1,I_2,\ldots,I_n\}$ of $G$ so that $L(I_i) < L(I_{i+1})$ (and $R(I_i) < R(I_{i+1})$)
for each $i$, $1 \le i \le n-1$, and each interval $I_i \in \calI$ corresponds to the vertex $v_i \in V$.

Let $\bfI_b = \{b_1, b_2, \ldots, b_k\}$ and $\bfI_r=\{r_1, r_2, \ldots, r_k\}$ be 
any given initial and target independent sets of $G$, respectively. 
Without loss of generality, we assume that the blue vertices $b_1, b_2, \ldots, b_k$ 
are labeled from left to right (according to the interval representation $\calI$ of $G$), 
that is, $L(b_i) < L(b_j)$ if $i < j$; similarly, 
we assume that the red vertices $r_1, r_2, \ldots, r_k$ are labeled from left to right. 
Then, we define a target-assignment $\fallmap: \bfI_b \to \bfI_r$, as follows:
for each blue vertex $b_i \in \bfI_b$
\begin{equation}
 \label{eq:map_proper}
 	\fallmap(b_i) = r_i.
\end{equation}
To prove Theorem~\ref{th:proper}, it suffices to show that $\fallmap$ is proper, 
and each token takes no detours.
	
\subsection{String representation}
By traversing the interval representation $\calI$ of a connected proper interval graph $G$ 
from left to right, we can obtain a string $S =s_1s_2\cdots s_{2k}$ which is a superstring of 
both $b_1b_2 \cdots b_k$ and $r_1r_2 \cdots r_k$, that is, 
each letter $s_i$ in $S$ is one of the vertices in $\bfI_b \cup \bfI_r$ and 
$s_i$ appears in $S$ before $s_j$ if $L(s_i) < L(s_j)$.
We may assume without loss of generality that $s_1=b_1$ 
since the reconfiguration rule is symmetric in {\sc sliding token}.
If a vertex is contained in both $\bfI_b$ and $\bfI_r$, 
as $b_i$ and $r_j$, then we assume that it appears as $b_i r_j$ in $S$, 
that is, the blue vertex $b_i$ appears in $S$ before the red vertex $r_j$.
Then, for each $i$, $1 \le i \le 2k$, we define the {\em height $h(i)$ at $i$} by 
the number of blue vertices appeared in the substring $s_1 s_2 \cdots s_i$ minus 
the number of red vertices appeared in $s_1 s_2 \cdots s_i$. 
For the sake of notational convenience, we define $h(0) = 0$.
Then $h(i)$ can be recursively computed as follows:
\begin{equation} 
\label{eq:height}
h(i) = \left\{ 
	\begin{array}{ll}
	0            & ~~~\mbox{if $i=0$}; \\
	h(i-1) + 1 & ~~~\mbox{if $s_i$ is blue}; \\
	h(i-1) - 1 & ~~~\mbox{if $s_i$ is red}.
	\end{array} \right.
\end{equation}
Note that $h(2k) = 0$ for any string $S$ since $\msize{\bfI_b} = \msize{\bfI_r}$.

Using the notion of height, we split the string $S$ into substrings $S_1, S_2, \ldots, S_h$ 
at every point of height $0$, that is, in each substring $S_j = s_{2p+1} s_{2p+2} \cdots s_{2q}$, 
we have $h(2q) = 0$ and $h(i) \neq 0$ for all $i$, $2p+1 \le i \le 2q-1$. 
For example, a string $S=b_1b_2r_1r_2 b_3r_3 r_4r_5b_4r_6b_5r_7b_6r_8b_7b_8 b_9r_9$ 
can be split into four substrings 
$S_1= b_1b_2r_1r_2$, $S_2=b_3r_3$, $S_3=r_4r_5b_4r_6b_5r_7b_6r_8b_7b_8$ and $S_4=b_9r_9$. 
Then, the substrings $S_1, S_2, \ldots, S_h$ form a partition of $S$, 
and each substring $S_j$ contains the same number of blue and red tokens.
We call such a partition the {\em partition of $S$ at height $0$}. 

\begin{lemma} \label{lem:balance}
Let $S_j = s_{2p+1} s_{2p+2} \cdots s_{2q}$ be a substring in the partition of the string $S$ at height $0$. Then,
\begin{listing}{aaa}
 \item[{\rm (}a{\rm )}] the blue vertices $b_{p+1}, b_{p+2}, \ldots, b_{q}$ appear in $S_j$, 
  and their corresponding red vertices $r_{p+1}, r_{p+2}, \ldots, r_{q}$ appear in $S_j${\rm ;} 
 \item[{\rm (}b{\rm )}] if $S_j$ starts with the blue vertex $b_{p+1}$, 
  then each blue vertex $b_i$, $p+1 \le i \le q$, 
  appears in $S_j$ before its corresponding red vertex $r_i${\rm ;} and
 \item[{\rm (}c{\rm )}] if $S_j$ starts with the red vertex $r_{p+1}$, 
  then each blue vertex $b_i$, $p+1 \le i \le q$, appears in $S_j$ after its corresponding red vertex $r_i$. 
\end{listing}
\end{lemma}
\begin{proof}
By the definitions, the claim (a) clearly holds.
We thus show that the claim (b) holds. (The proof for the claim (c) is symmetric.)

Since $h(2p) = 0$ and $S_j$ starts with a blue vertex, we have $h(2p+1) = 1 > 0$. 
We now suppose for a contradiction that there exists 
a blue vertex $s_x = b_{i'}$ which appears in $S_j$ after its corresponding red vertex $s_y = r_{i'}$. 
Then, $y < x$. 
We assume that $y$ is the minimum index among such blue vertices in $S_j$.
Then, in the substring $s_1 s_2 \cdots s_y$ of $S$, there are exactly $i'$ red vertices. 
On the other hand, since $y < x$, the substring $s_1 s_2 \cdots s_y$ contains at most $i'-1$ blue vertices. 
Therefore, by the definition of height, we have $h(y) < 0$. 
Since $h(2p+1) = 1 > 0$ and  $h(y) < 0$, by Eq.~(\ref{eq:height}) there must exist an index $z$ 
such that $h(z) = 0$ and $2p < z < y$. 
This contradicts the fact that $S_j$ is a substring in the partition of $S$ at height $0$. 
\qed
\end{proof}

\subsection{Algorithm}
\label{subsec:algo_proper}
	
Recall that we have fixed the unique interval representation $\calI =\{I_1,I_2,\ldots,I_n\}$ of 
a connected proper interval graph $G$ so that $L(I_i) < L(I_{i+1})$ for each $i$, 
$1 \le i \le n-1$, and each interval $I_i \in \calI$ corresponds to the vertex $v_i \in V$.
Since all intervals in $\calI$ have unit length, the following proposition clearly holds.
\begin{proposition}
\label{prop:optimalway}
For two vertices $v_i$ and $v_j$ in $G$ such that $i < j$, 
there is a path $P$ in $G$ which passes through only intervals {\rm (}vertices{\rm )} 
contained in $[L(I_{i}), R(I_{j})]$.
Furthermore, if $I_{i'} \cap I_i = \emptyset$ for some index $i'$ with $i'<i$, 
no vertex in $v_1, v_2, \ldots, v_{i'}$ is adjacent to any vertex in $P$. 
If $I_{j} \cap I_{j'} = \emptyset$ for some index $j'$ with $j<j'$, 
no vertex in $v_{j'}, v_{j'+1}, \ldots, v_{n}$ is adjacent to any vertex in $P$. 
\end{proposition}

Let $S$ be the string of length $2k$ obtained from two given independent sets $\bfI_b$ and $\bfI_r$ 
of a connected proper interval graph $G$, where $k = \msize{\bfI_b} = \msize{\bfI_r}$. 
Let $S_1, S_2, \ldots, S_h$ be the partition of $S$ at height $0$. 
The following lemma shows that the tokens in each substring $S_j$ 
can always reach their corresponding red vertices.
(Note that we sometimes denote simply by $S_j$ the set of all 
vertices appeared in the substring $S_j$, $1 \le j \le h$.) 
\begin{lemma}
\label{lem:proper-simple}
Let $S_j = s_{2p+1} s_{2p+2} \cdots s_{2q}$ be a substring in the partition of $S$ at height $0$. 
Then, there exists a reconfiguration sequence between $\bfI_b \cap S_j$ and $\bfI_r \cap S_j$ 
such that tokens are slid along edges only in the subgraph of $G$ 
induced by the vertices contained in $[L(s_{2p+1}), R(s_{2q})]$. 
\end{lemma}
\begin{proof}
We first consider the case where $S_j$ starts with the blue vertex $b_{p+1}$, that is, $s_{2p+1} = b_{p+1}$. 
Then, by Lemma~\ref{lem:balance}(b) each blue vertex $b_i$, $p+1 \le i \le q$, 
appears in $S_j$ before the corresponding red vertex $r_i$.
Therefore, we know that $s_{2q} = r_q$, and hence it is red.
Suppose that $s_{x} = b_{q}$, then all vertices appeared in $s_{x+1} s_{x+2} \cdots s_{2q}$ are red. 
Roughly speaking, we slide the tokens $t_q, t_{q-1}, \ldots, t_{p+1}$ from left to right in this order. 

We first claim that the token $t_q$ can be slid from $b_q$ $(= s_x)$ to $r_q$ $(= s_{2q})$. 
By Proposition~\ref{prop:optimalway} 
there is a path $P$ between $b_q$ and $r_q$ which passes through only intervals contained in $[L(b_q), R(r_q)]$. 
Since $\bfI_b$ is an independent set of $G$, 
the vertex $b_q$ is not adjacent to any other vertices $b_{p+1}, b_{p+2}, \ldots, b_{q-1}$ in $\bfI_b \cap S_j$. 
Since $L(b_{p+1}) < L(b_{p+2}) < \cdots < L(b_{q-1}) < L(b_q)$, 
by Proposition~\ref{prop:optimalway} all vertices in $P$ are not adjacent to 
any of tokens $t_{p+1}, t_{p+2}, \ldots, t_{q-1}$ that are now placed on 
$b_{p+1}, b_{p+2}, \ldots, b_{q-1}$, respectively.  
Therefore, we can slide the token $t_q$ from $b_q$ to $r_q$. 
We fix the token $t_q$ on $r_q=s_{2q}$, and will not slide it anymore.

We then slide the next token $t_{q-1}$ on $b_{q-1}$ to $r_{q-1}$ along 
a path $P'$ which passes through only intervals contained in $[L(b_{q-1}), R(r_{q-1})]$. 
Since $\bfI_r$ is an independent set of $G$, 
the corresponding red vertex $r_{q-1}$ is not adjacent to $r_q$ on which the token $t_q$ is now placed. 
Recall that $L(r_{q-1}) < L(r_q)$, and hence by Proposition~\ref{prop:optimalway}, 
$r_q$ is not adjacent to any vertex in $P'$. 
Similarly as above, the tokens $t_{p+1}, t_{p+2}, \ldots, t_{q-2}$ are not adjacent to any vertex in $P'$. 
Therefore, we can slide the token $t_{q-1}$ from $b_{q-1}$ to $r_{q-1}$. 

Repeat this process until the token $t_{p+1}$ on $b_{p+1}$ is slid to $r_{p+1}$. 
In this way, there is a reconfiguration sequence between $\bfI_b \cap S_j$ and $\bfI_r \cap S_j$ 
such that tokens are slid along edges only in the subgraph of $G$ induced by 
the vertices contained in $[L(b_{p+1}), R(r_{q})]$. 

The symmetric arguments prove the case where $S_j$ starts with the red vertex $r_{p+1}$. 
Note that, in this case, we slide the tokens $t_{p+1}, t_{p+2}, \ldots, t_q$ from right to left in this order. 
\qed
\end{proof}

	
\noindent
{\bf Proof of Theorem~\ref{th:proper}.}
We now give an algorithm which slides all tokens on the vertices in $\bfI_{b}$ to the vertices in $\bfI_r$.
Recall that $S_1, S_2, \ldots, S_h$ are the substrings in the partition of $S$ at height $0$. 
Intuitively, the algorithm repeatedly picks up one substring $S_j$, 
and slides all tokens in $\bfI_b \cap S_j$ to $\bfI_r \cap S_j$. 
By Lemma \ref{lem:proper-simple} it works locally in each substring $S_j$, 
but it should be noted that a token in $S_j$ may be adjacent to another token 
in $S_{j-1}$ or $S_{j+1}$ at the boundary of the substrings.
To avoid this, we define a partial order over the substrings $S_1, S_2, \ldots, S_h$, as follows.

Consider any two consecutive substrings $S_j$ and $S_{j+1}$, 
and let $S_j = s_{2p+1} s_{2p+2} \allowbreak \cdots s_{2q}$.
Then, the first letter of $S_{j+1}$ is $s_{2q+1}$. 
We first consider the case where both $s_{2q}$ and $s_{2q+1}$ are the same color.
Then, since $s_{2q}$ and $s_{2q+1}$ are both in the same independent set of $G$, 
they are not adjacent.
Therefore, by Proposition~\ref{prop:optimalway} and Lemma~\ref{lem:proper-simple}, 
we can deal with $S_j$ and $S_{j+1}$ independently.
In this case, we thus do not define the ordering between $S_j$ and $S_{j+1}$.
We then consider the case where $s_{2q}$ and $s_{2q+1}$ have different colors; 
in this case, we have to define their ordering.
Suppose that $s_{2q}$ is blue and $s_{2q+1}$ is red;
then we have $s_{2q} = b_q$ and $s_{2q+1} = r_{q+1}$.
By Lemma~\ref{lem:proper-simple} the token $t_q$ on $s_{2q}$ is slid to left, 
and the token $t_{q+1}$ will reach $r_{q+1}$ from right. 
Therefore, the algorithm has to deal with $S_j$ before $S_{j+1}$.
Note that, after sliding all tokens $t_{p+1}, t_{p+2}, \ldots, t_q$ in $S_j$, 
they are on the red vertices $r_{p+1}, r_{p+2}, \ldots, r_q$, respectively, 
and hence the tokens in $S_{j+1}$ are not adjacent to any of them.
By the symmetric argument, if $s_{2q}$ is red and $s_{2q+1}$ is blue, 
$S_{j+1}$ should be dealt with before $S_j$.

Notice that such an ordering is defined only for two consecutive substrings $S_j$ and $S_{j+1}$, $1 \le j \le h-1$. 
Therefore, the partial order over the substrings $S_1, S_2, \ldots, S_h$ is acyclic, 
and hence there exists a total order which is consistent with the partial order defined above.
The algorithm certainly slides all tokens from $\bfI_b$ to $\bfI_r$ according to the total order.
Therefore, the target-assignment $\fallmap$ defined in Eq.~(\ref{eq:map_proper}) is proper, 
and hence $\bfI_{b}\vdash^* \bfI_{r}$.


	

Therefore, there always exists a reconfiguration sequence 
between two independent sets $\bfI_b$ and $\bfI_r$ of a connected proper interval graph $G$. 
We now discuss the length of reconfiguration sequences between $\bfI_b$ and $\bfI_r$, 
together with the running time of our algorithm. 
\begin{proposition} 
\label{prop:proper}
For two given independent sets $\bfI_b$ and $\bfI_r$ of 
a connected proper interval graph $G$ with $n$ vertices, 
\begin{listing}{aaa}
 \item[$(1)$] the ordering of tokens to be slid in a shortest reconfiguration sequence between 
	      them can be computed in $O(n)$ time and $O(n)$ space{\rm ;} and 
 \item[$(2)$] a shortest reconfiguration sequence between them can be output 
	      in $O(n^2)$ time and $O(n)$ space.
\end{listing}
\end{proposition}
\begin{proof}
We first modify our algorithm 
so that it finds a shortest reconfiguration sequence between $\bfI_b$ and $\bfI_r$. 
To do that, it suffices to slide each token $t_i$, $1 \le i \le k$, 
from the blue vertex $b_i$ to its corresponding red vertex $r_i$ along the shortest path between $b_i$ and $r_j$.  
We may assume without loss of generality that $L(b_i) < L(r_i)$, 
that is, the token $t_i$ will be slid from left to right.
Then, for the interval $b_i$, we choose an interval $I_{j} \in \calI$ 
such that $b_i \cap I_{j} \neq \emptyset$ and $L(I_j)$ is the maximum among all $I_{j'} \in \calI$. 
If $L(r_i)\le L(I_j)$, we can slide $t_i$ from $b_i$ to $r_i$ directly;
otherwise we slide $t_i$ to the vertex $I_j$, and repeat.
	
We then prove the claim (1).
If we simply want to compute the ordering of tokens to be slid in a shortest reconfiguration sequence, 
it suffices to compute the partial order over the substrings $S_1, S_2, \ldots, S_h$ 
in the partition of the string $S$ at height $0$. 
It is not difficult to implement our algorithm in Section~\ref{subsec:algo_proper} 
to run in $O(n)$ time and $O(n)$ space.
Therefore, the claim (1) holds. 
	
We finally prove the claim (2).
Remember that each token $t_i$, $1 \le i \le k$, is slid along the shortest path from $b_i$ to $r_i$. 
Furthermore, once the token $t_i$ reaches $r_i$, it is not slid anymore. 
Therefore, the length of a shortest reconfiguration sequence 
between $\bfI_{b}$ and $\bfI_{r}$ is given by 
the sum of all lengths of the shortest paths between $b_i$ and $r_i$.
It is clear that this sum is $O(kn)=O(n^2)$.
We output only the shortest paths between $b_i$ and $r_i$, 
together with the ordering of the tokens to be slid. 
Therefore, the claim (2) holds.
\qed
\end{proof}

This proposition also completes the proof of Theorem~\ref{th:proper}.
\qed

It is remarkable that there exists an infinite family of instances 
for which any reconfiguration sequence requires $\Omega(n^2)$ length. 
To show this, we give an instance such that each shortest path between $b_i$ and $r_i$ is $\Theta(n)$.
Simple example is: $G$ is a path $(v_1,v_2,\ldots,v_{8k})$ of length $n=8k$ for any positive integer $k$,
$\bfI_b=\{v_1,v_3,v_5,\ldots,v_{2k-1}\}$, and $\bfI_r=\{v_{6k+2},v_{6k+4},\ldots,v_{8k}\}$.
In this instance, each token $t_i$ must be slid $\Theta(n)$ times, 
and hence it requires $\Theta(n^2)$ time to output all of them.
We note that a path is not only a proper interval graph, but also a caterpillar.
Thus this simple example also works as a caterpillar.


\section{Trivially perfect graphs}
\label{sec:trivially}

The main result of this section is the following theorem.
\begin{theorem}
\label{th:trivialperfect}
The {\sc sliding token} problem for 
a trivially perfect graph $G = (V,E)$ can be solved in $O(n)$ time and $O(n)$ space.
Furthermore, one can find a shortest reconfiguration sequence between 
two given independent sets $\bfI_b$ and $\bfI_r$ in $O(n)$ time and $O(n)$ space if there exists. 
\end{theorem}
	
In this section, we explicitly give such an algorithm as a proof of Theorem~\ref{th:trivialperfect}.
Note that there are no-instances for trivially perfect graphs.
However, for trivially perfect graphs, we construct
a proper target-assignment between $\bfI_b$ and $\bfI_r$ efficiently if it exists.


\subsection{$\MPQ$-tree for trivially perfect graphs}

The $\MPQ$-tree of an interval graph $G$ is a kind of decomposition tree, 
developed by Korte and M\"ohring \cite{KM89}, 
which represents the set of all feasible interval representations of $G$. 
For an interval graph $G$, 
although there are exponentially many interval representations for $G$,
its corresponding $\MPQ$-tree is unique up to isomorphism.
%
For the notion of $\MPQ$-trees, the following theorem is known:
\begin{theorem}[{\cite{KM89}}]\label{th:MPQ} 
For any interval graph $G=(V,E)$, its corresponding $\MPQ$-tree can be constructed in $O(n+m)$ time.
\end{theorem}

Since it is involved to define $\MPQ$-tree for general interval graphs,
we here give a simplified definition of $\MPQ$-tree only for the class of trivially perfect graphs.
(See \cite{KM89} for the detailed definition of the $\MPQ$-tree for a general interval graph.)
%
Let $G=(V,E)$ be a trivially perfect graph. Recall that a trivially perfect graph has 
an interval representation such that the relationship between any two intervals is either disjoint or inclusion.
Then, the {\em $\MPQ$-tree} $\calT$ of $G$ is a rooted tree such that each node, 
called a {\em $\PP$-node}, in $\calT$ is associated with a non-empty set of vertices in $G$ such that
(a) each vertex $v\in V$ appears in exactly one $\PP$-node in $\calT$, and
(b) if a vertex $v_i \in V$ is in an ancestor node of another node that contains $v_j \in V$, 
then $L(I_i)\le L(I_j)<R(I_j)\le R(I_i)$ in any interval representation of $G$, 
where $v_i$ and $v_j$ correspond to the intervals $I_i$ and $I_j$, respectively (see \figurename~\ref{fig:mpq} as an example).
%
%
By the property (b), the ancestor/descendant relationship on $\calT$ corresponds to 
the inclusion relationship in the interval representation of $G$.
Thus, $N[v_j]\subseteq N[v_i]$ if $v_i$ is in an ancestor node of another node that contains $v_j$ in the $\MPQ$-tree.
%
%

Let $\calT$ be the (unique) $\MPQ$-tree of a connected trivially perfect graph $G = (V,E)$. 
For two vertices $u$ and $w$ in $G$, 
we denote by $\LCA(u,w)$ the least common ancestor in $\calT$ for the nodes containing $u$ and $w$. 
By the property (a) the node $\LCA(u,w)$ can be uniquely defined.

\begin{figure}[t]
\begin{center}
\includegraphics[width=0.6\textwidth]{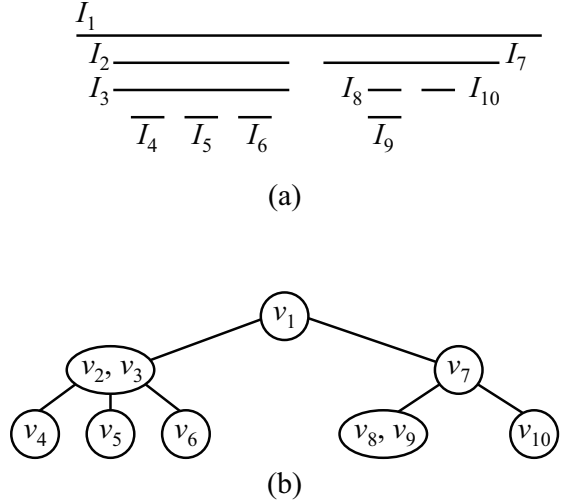}
\end{center}
\caption{(a) A trivially perfect graph in an interval representation, and (b) its $\MPQ$-tree.}
\label{fig:mpq}
\end{figure}

\subsection{Basic properties and key lemma}
\label{subsec:keylemma}
	
Let $\calT$ be the (unique) $\MPQ$-tree of a connected trivially perfect graph $G = (V,E)$. 
Recall that the interval representation of a trivially perfect graph has just disjoint or inclusion relationship.
This fact implies the following observation. 
\begin{observation}
\label{obs:twosteps}
Every pair of vertices $u$ and $w$ in a connected trivially perfect graph $G$ has 
a path of length at most two via a vertex in $\LCA(u,w)$. 
\end{observation}
\begin{proof}
We first observe that every $\PP$-node of $\calT$ is non-empty.
If the root is empty, the graph is disconnected, which is a contradiction.
If some non-root $\PP$-node $P$ is empty, joining all children of $P$ to the parent of $P$,
we obtain a simpler $\MPQ$-tree than $\calT$, which contradicts the construction of 
the unique $\MPQ$-tree in \cite{KM89}. Thus every $\PP$-node is non-empty.
Therefore, there is at least one vertex $v$ in $\LCA(u,w)$.
Then, the property (b) implies that $N[u]\subseteq N[v]$ and $N[w]\subseteq N[v]$, 
and hence $\{u,v\}$ and $\{w,v\}$ are both in $E$.
Therefore, there is a path $(u,v,w)$ of length at most two between $u$ and $w$ via $v$.
When we have $v=u$ or $v = w$, the path degenerates to the edge $\{u,w\} \in E$.
\qed
\end{proof}

Let $\LCA^*(u,w)$ be the set of vertices in $V$ appearing in 
the $\PP$-nodes on the (unique) path from $\LCA(u,w)$ to the root of the $\MPQ$-tree.
By the definition of $\MPQ$-tree, we clearly have the following observation. 
(Recall also that each token must be slid along an edge of $G$.)
\begin{observation}
\label{obs:lca_onetoken}
Consider an arbitrary reconfiguration sequence $\calS$ 
which slides a token $t_i$ from $b_i \in \bfI_b$ to some vertex $r_i$.
Then, $t_i$ must pass through at least one vertex in $\LCA^*(b_i, r_i)$, 
that is, there exists at least one independent set $\bfI'$ in $\calS$ 
such that $\bfI' \cap \LCA^*(b_i, r_i) \neq \emptyset$.
\end{observation}

We are now ready to give the key lemma for trivially perfect graphs.
\begin{lemma}
\label{lem:lca}
Let $\fallmap: \bfI_b \to \bfI_r$ be a target-assignment between $\bfI_b$ and $\bfI_r$. 
Then, $\fallmap$ is proper if and only if
the nodes $\LCA(b_i, \fallmap(b_i))$ and $\LCA(b_j, \fallmap(b_j))$ are not in 
the ancestor/descendant relationship on $\calT$ for every pair of vertices $b_i, b_j \in \bfI_b$.
\end{lemma}
\begin{proof}
We first show the sufficiency. 
For a target-assignment $\fallmap$ between $\bfI_b$ and $\bfI_r$, 
suppose that the nodes $\LCA(b_i, \fallmap(b_i))$ and $\LCA(b_j, \fallmap(b_j))$ 
are not in the ancestor/descendant relationship on $\calT$ for every pair of vertices $b_i, b_j \in \bfI_b$. 
Then, we can simply slide the tokens one by one in an arbitrary order; 
by Observation~\ref{obs:twosteps} each token $t_i$, $1 \le i \le k$, 
can be slid along a path from $b_i$ to $\fallmap(b_i)$ via a vertex $v_{i'}$ in $\LCA(b_i, \fallmap(b_i))$. 
Note that there is no token $t_j$ adjacent to $v_{i'}$, 
because the nodes $\LCA(b_i, \fallmap(b_i))$ and $\LCA(b_j, \fallmap(b_j))$ are not in 
the ancestor/descendant relationship on $\calT$. 
Thus, there is a reconfiguration sequence between $\bfI_b$ and $\bfI_r$ according to $\fallmap$, 
and hence $\fallmap$ is proper.   
	
We then show the necessity. 
Suppose that $\fallmap$ is proper, 
and suppose for a contradiction that 
there exists a pair of vertices $b_i, b_j \in \bfI_b$ such that 
the nodes $\LCA(b_i, \fallmap(b_i))$ and $\LCA(b_j, \fallmap(b_j))$ are in 
the ancestor/descendant relationship on $\calT$; 
without loss of generality, we assume that $\LCA(b_i, \fallmap(b_i))$ is 
an ancestor of $\LCA(b_j, \fallmap(b_j))$.
Since $\fallmap$ is proper, 
there exists a reconfiguration sequence $\calS$ between $\bfI_b$ and $\bfI_r$ 
which slides the token $t_i$ from $b_i$ to $\fallmap(b_i)$ and also slides 
the token $t_j$ from $b_j$ to $\fallmap(b_j)$. 
By Observation~\ref{obs:lca_onetoken} 
there is at least one vertex $v_{i'}$ in $\LCA^*(b_i, \fallmap(b_i))$ which is passed through by $t_i$. 
Similarly, there is at least one vertex $v_{j'}$ in $\LCA^*(b_j, \fallmap(b_j))$ 
which is passed through by $t_j$.  
Let $P_i$ and $P_j$ be the $\PP$-nodes that contains $v_{i'}$ and $v_{j'}$, respectively.
Since $\LCA(b_i, \fallmap(b_i))$ and $\LCA(b_j, \fallmap(b_j))$ are in 
the ancestor/descendant relationship on $\calT$, so are $P_i$ and $P_j$.
First suppose $P_j$ is an ancestor of $P_i$.
Then we have 
$N[b_{j}] \subseteq N[v_{i'}]$, $N[\fallmap(b_j)]\subseteq N[v_{i'}]$ and $N[v_{j'}]\subseteq N[v_{i'}]$.
Therefore, if we slide $t_i$ via $v_{i'}$, then $t_i$ would be adjacent to 
the other token $t_j$ which is on one of the three vertices $b_{j}$, $\fallmap(b_j)$ and $v_{j'}$.
Thus, the token $t_j$ should ``escape'' from $b_j$ before sliding $t_i$.
However, we can establish the same argument for any descendant of $P_i$, 
and hence $t_j$ must escape to some vertex $u$ that is contained in an ancestor of $P_i$ at first.
However, the vertex $u$ is adjacent to all of $b_{i}, \fallmap(b_i), v_{i'}$, 
and hence $t_j$ cannot escape before sliding $t_i$.
This contradicts the assumption that $\calS$ slides the token $t_i$ 
from $b_i$ to $\fallmap(b_i)$ and also slides the token $t_j$ from $b_j$ to $\fallmap(b_j)$. 
The other case, $P_i$ is an ancestor of $P_i$, is symmetric.
\qed
\end{proof}

\subsection{Algorithm and its correctness}
\label{subsec:algo_trivially}
	
We now describe our linear-time algorithm for a trivially perfect graph.
Let $\calT$ be the $\MPQ$-tree of a connected trivially perfect graph $G = (V,E)$.
Let $\bfI_b=\{b_1,b_2,\ldots,b_k\}$ and $\bfI_r=\{r_1,r_2,\ldots,r_k\}$ be given 
initial and target independent sets of $G$, respectively. 
Then, we determine whether $\bfI_b \vdash^* \bfI_r$ as follows:
\begin{listing}{aaa}
 \item[(A)] construct some particular target-assignment $\fallmap^*$ 
	    between $\bfI_b$ and $\bfI_r$; and 
 \item[(B)] check whether $\fallmap^*$ is proper or not, using Lemma~\ref{lem:lca}. 
\end{listing}
We will show later in Lemma~\ref{lem:no_case} that 
it suffices to check only $\fallmap^*$ in order to determine 
whether $\bfI_b \vdash^* \bfI_r$ or not. 
Indeed, our linear-time algorithm executes (A) and (B) above at the same time, 
in the bottom-up manner based on $\calT$.

\paragraph{Description of the algorithm}

Remember that the vertex-set associated to each $\PP$-node in $\calT$ induces a clique in $G$. 
Therefore, for any independent set $\bfI$ of $G$, 
each $\PP$-node contains at most one vertex in $\bfI$, 
and hence contains at most one token. 
We put a ``blue token'' for each $\PP$-node containing a blue vertex in $\bfI_b$, 
and also put a ``red token'' for each $\PP$-node containing a red vertex in $\bfI_r$.
Note that a $\PP$-node may contain a pair of blue and red tokens. 
Our algorithm lifts up the tokens from the leaves to the root of $\calT$, 
and if a blue token $b$ meets a red token $r$ at their least common ancestor 
$\LCA(b,r)$ in $\calT$, then we replace them by a single ``green token.''
This corresponds to setting $\fallmap^*(b) = r$.
More precisely, at initialization step, 
the algorithm first collects all leaves of $\calT$ in a queue, which is called {\em frontier}.
The algorithm marks the nodes in the frontier, 
and lifts up each token to its parent $\PP$-node.
Each $\PP$-node $P$ is put into the frontier if its all children are marked,
and then, all children of $P$ are removed from the frontier after 
the following procedure at $P$:
\begin{listing}{aaa}
 \item[Case (1)] $P$ contains at most one token: the algorithm has nothing to do.
 \item[Case (2)] $P$ contains only one pair of blue token $b$ and red token $r$: 
	   the algorithm replaces them by a single green token, and let $\fallmap^*(b) = r$.
 \item[Case (3)] $P$ contains only green tokens: 
	    the algorithm replaces them by a single green token.
 \item[Case (4)] $P$ contains two or more blue tokens, or two or more red tokens: 
            the algorithm outputs ``no'' and halts
	    (that is, $\bfI_b\not\vdash^* \bfI_r$ in this case).
 \item[Case (5)] $P$ contains at least one green token and at least one blue or red token: 
	    the algorithm outputs ``no'' and halts
	    (that is, $\bfI_b\not\vdash^* \bfI_r$ in this case).
\end{listing}
Repeating this process, and the algorithm outputs ``yes'' 
if and only when the frontier contains only 
the root $\PP$-node $r$ of $\calT$ which is in one of Cases (1)--(3) above.

\paragraph{Correctness of the algorithm}

It is not difficult to implement our algorithm to run in $O(n)$ time and $O(n)$ space.
Therefore, we here prove the correctness of the algorithm.

We first show that $\bfI_b \vdash^* \bfI_r$ if the algorithm outputs ``yes.'' 
In this case, the algorithm is in Cases (1), (2), or (3) 
at each $\PP$-nodes in $\calT$ (including the root $r$).
Then, the target-assignment $\fallmap^*$ has been (completely) constructed: 
for each blue vertex $b_i \in \bfI_b$, 
$\fallmap^*(b_i)$ is the red vertex in $\bfI_r$ such that 
$\LCA(b_i, v_{i'})$ has the minimum height in $\calT$ among 
all vertices $v_{i'} \in \bfI_r$. 
Then, we have the following lemma. 
\begin{lemma} 
\label{lem:yes_case}
If the algorithm outputs ``yes,'' then $\bfI_b \vdash^* \bfI_r$. 
\end{lemma}
\begin{proof}
By Lemma~\ref{lem:lca} it suffices to show that 
the target-assignment $\fallmap^*$ constructed by 
the algorithm satisfies that the nodes $\LCA(b_i, \fallmap^*(b_i))$ and
$\LCA(b_j, \fallmap^*(b_j))$ are not in the ancestor/descendant relationship 
on $\calT$ for every pair of vertices $b_i, b_j \in \bfI_b$.

We first consider the case where a $\PP$-node $P$ is in Case (2). 
Then, there is exactly one pair of a blue token $b$ and 
a red token $r = \fallmap^*(b)$, and $P=\LCA(b,r)$. 
Since $b$ and $r$ did not meet any other tokens before $P$, 
the subtree $\calT_P$ of $\calT$ contains only the two tokens $b$ and $r$. 
Therefore, the lemma clearly holds for $\calT_P$.

We then consider the case where a $\PP$-node $P$ is in Case (3). 
Then, two or more least common ancestors of pairs of blue and red tokens meet at this node $P$. 
Notice that the green tokens were placed on children's node of $P$ in 
the previous step of the algorithm, and hence they were sibling in $\calT$.
Therefore, their corresponding least common ancestors are not in 
the ancestor/descendant relationship on $\calT$. 
\qed
\end{proof}

The following lemma completes the correctness proof of our algorithm. 
\begin{lemma} 
\label{lem:no_case}
If the algorithm outputs ``no,'' then $\bfI_b \not\vdash^* \bfI_r$. 
\end{lemma}
\begin{proof}
We assume that the algorithm outputs ``no.''
Then, by Lemma~\ref{lem:lca}, it suffices to show that 
there is no target-assignment $\fallmap$ between 
$\bfI_b$ and $\bfI_r$ such that $\LCA(b_i, \fallmap(b_i))$ and 
$\LCA(b_j, \fallmap(b_j))$ are not in the ancestor/descendant 
relationship on $\calT$ for every pair of vertices $b_i, b_j \in \bfI_b$.

Suppose for a contradiction that 
the algorithm outputs ``no,'' but there exists 
a target-assignment $\fallmap'$ between $\bfI_b$ and $\bfI_r$ 
such that $\LCA(b_i, \fallmap'(b_i))$ and $\LCA(b_j, \fallmap'(b_j))$ 
are not in the ancestor/descendant relationship on $\calT$ 
for every pair of vertices $b_i, b_j \in \bfI_b$.
Then, by Lemma~\ref{lem:lca}, $\fallmap'$ is proper and hence $\bfI_b\vdash^* \bfI_r$.
Since the algorithm outputs ``no,'' 
there is a $\PP$-node $P$ which is in either Case (4) or (5). 
	
We first assume that the $\PP$-node $P$ is in Case (4).
Without loss of generality, 
at least two blue tokens $b_1$ and $b_2$ meet at this node $P$.
Then, the $\MPQ$-tree $\calT$ contains two red tokens $r_1$ and $r_2$ 
placed on $\fallmap'(b_1)$ and $\fallmap'(b_2)$, respectively.
Notice that, since $b_1$ and $b_2$ did not meet any red token before 
at the node $P$, 
both $r_1$ and $r_2$ must be placed on either $P$ or nodes in $\calT \setminus \calT_P$.
Then, the least common ancestor $\LCA(b_1, \fallmap'(b_1))$ must be an ancestor of $P$, 
and so is the least common ancestor $\LCA(b_2, \fallmap'(b_2))$.
Therefore, the nodes $\LCA(b_1, \fallmap'(b_1))$ and $\LCA(b_2, \fallmap'(b_2))$ 
are in the ancestor/descendant relationship, a contradiction. 

Thus, the algorithm outputs ``no'' because the $\PP$-node $P$ is in Case (5). 
In this case, without loss of generality, 
at least one blue token $b_1$ and at least one green token $c_1$ meet at this node $P$.
Then, the red token $r_1$ corresponding to $\fallmap'(b_1)$ must be placed on 
either $P$ or some node in $\calT \setminus \calT_P$.
Therefore, the least common ancestor $\LCA(b_1, \fallmap'(b_1))$ is an ancestor of $c_1$.
Note that $c_1$ corresponds to the least common ancestor of some pair of blue and red tokens, 
say $b_j$ and $\fallmap'(b_j)$, and $p$ is an ancestor of it.
Therefore, the nodes $\LCA(b_1, \fallmap'(b_1))$ and $\LCA(b_j, \fallmap'(b_j))$ are 
in the ancestor/descendant relationship on $\calT$, a contradiction. 
\qed
\end{proof}

\subsection{Shortest reconfiguration sequence}

To complete the proof of Theorem~\ref{th:trivialperfect}, 
we finally show that our algorithm in Section~\ref{subsec:algo_trivially} 
can be modified so that it actually finds a shortest reconfiguration sequence between $\bfI_b$ and $\bfI_r$. 

Once we know that $\bfI_b \vdash^* \bfI_r$ holds by the $O(n)$-time algorithm 
in Section~\ref{subsec:algo_trivially}, 
we run it again with modification that ``green'' tokens are left at 
the corresponding least common ancestors.
As in the proof of Lemma~\ref{lem:lca}, we can now obtain a reconfiguration sequence 
$\calS = \langle \bfI_1, \bfI_2, \ldots, \bfI_{\ell} \rangle$ 
between $\bfI_b = \bfI_1$ and $\bfI_r = \bfI_{\ell}$ such that 
each token $t_i$, $1 \le i \le k$, is slid at most twice.
It is sufficient to output $\bfI_{i+1}\setminus \bfI_i$ and $\bfI_{i}\setminus \bfI_{i+1}$, 
and hence the running time of the modified algorithm is proportional to $\ell$, 
the number of independent sets in $\calS$. 
Since $k=\msize{\bfI_b}=O(n)$ and each token $t_i$, $1 \le i \le k$, 
is slid at most twice in $\calS$, 
we have $\ell = O(n)$, that is, the length $\ell$ of $\calS$ is $O(n)$.
Therefore, the modified algorithm also runs in $O(n)$ time and $O(n)$ space.
Notice that each token $t_i$ is slid to its target vertex 
$\fallmap^*(b_i)$ along a shortest path (of length at most two) 
between $b_i$ and $\fallmap^*(b_i)$ without detour, and hence $\calS$ has the minimum length. 

This completes the proof of Theorem~\ref{th:trivialperfect}.

\section{Caterpillars}
\label{sec:caterpillars}

The main result of this section is the following theorem. 
\begin{theorem}
\label{th:caterpillar}
The {\sc sliding token} problem for a connected caterpillar $G = (V,E)$ 
and two independent sets $\bfI_b$ and $\bfI_r$ of $G$ can be solved in $O(n)$ time and $O(n)$ space.
Moreover, for a yes-instance, 
a shortest reconfiguration sequence between them can be output in $O(n^2)$ time and $O(n)$ space.
\end{theorem}

Let $G=(S\cup L,E)$ be a caterpillar with spine $S$ which induces 
the path $(s_1,\ldots,s_{n'})$, and leaf set $L$.
We assume that ${n'}\ge 2$, $\deg(s_1)\ge 2$, and $\deg(s_{n'})\ge 2$. 
First we show that we can assume that each spine vertex has at most one leaf without loss of generality.

\begin{lemma}
\label{lem:cat}
For any given caterpillar $G=(S\cup L,E)$ and two independent sets $\bfI_b$ and $\bfI_r$ on $G$,
there is a linear time reduction from them to another caterpillar 
$G'=(S'\cup L',E')$ and two independent sets $\bfI_b'$ and $\bfI_r'$ such that 
(1) $G$, $\bfI_b$, and $\bfI_r$ are a yes-instance of the {\sc sliding token} problem 
if and only if $G'$, $\bfI_b'$, and $\bfI_r$ are a yes-instance of the {\sc sliding token} problem,
(2) the maximum degree of $G'$ is at most 3, and 
(3) $\deg(s_1)=\deg(s_{n'})=2$, where $n'=\msize{S'}$.
In other words, the {\sc sliding token} problem on a caterpillar 
is sufficient to consider only caterpillars of maximum degree 3.
\end{lemma}
\begin{proof}
On $G$, let $s_i$ be any vertex in $S$ with $\deg(s_i)>3$.
Then there exist at least two leaves $\ell_i$ and $\ell_i'$ attached to $s_i$
(note that they are weak twins).
Now we consider the case that two tokens in $\bfI_b$ are on $\ell_i$ and $\ell_i'$.
Then, we cannot slide these two tokens at all, and any other token cannot pass through $s_i$ since 
 it is blocked by them.
If $\bfI_r$ contains these two tokens also, 
we can split the problem into two subproblems by removing $s_i$ and 
its leaves from $G$, and solve it separately.
Otherwise, the answer is ``no'' (remind that the problem is reversible; 
that is, if tokens cannot be slid, there are no other tokens which slide into the situation).
Therefore, if at least two tokens are placed on the leaves of 
 a vertex of the original graph, we can reduce the case in linear time.
Thus we assume that every spine vertex with its leaves contains 
at most one token in $\bfI_b$ and $\bfI_r$, respectively.
Then, by the same reason, we can remove all leaves but one of each spine vertex.
More precisely, regardless whether $\bfI_b \vdash^* \bfI_r$ or $\bfI_b \not\vdash^* \bfI_r$,
at most one leaf for each spine vertex is used for the transitions.
Therefore, we can remove all other useless leaves but one from each spine vertex.
Especially, removing all useless leaves, we have $\deg(s_1)=\deg(s_{n'})=2$.
\qed
\end{proof}

Hereafter, we only consider the caterpillars stated in Lemma \ref{lem:cat}.
That is, for any given caterpillar $G=(S\cup L,E)$ with spine $(s_1,\ldots,s_{n'})$,
we assume that $\deg(s_1)=\deg(s_{n'})=2$ and $2\le \deg(s_i)\le 3$ for each $1<i<n'$.
Then, we denote the unique leaf of $s_i$ by $\ell_i$ if it exists.

We here introduce a key notion of the problem on these caterpillars that is named {\em locked} path.
Let $G$ and $\bfI$ be a caterpillar and an independent set of $G$, respectively.
A path $P=(p_1,p_2,\ldots,p_k)$ on $G$ is {\em locked} by $\bfI$ if and only if
\begin{listing}{aaa}
\item[(a)] $k$ is odd and greater than 2,
\item[(b)] $\bfI \cap P=\{p_1,p_3,p_5,\ldots,p_k\}$,
\item[(c)] $\deg(p_1)=\deg(p_k)=1$ (in other words, they are leaves), and
	   $\deg(p_3)=\deg(p_5)=\cdots =\deg(p_{k-2})=2$.
\end{listing}
This notion is simplified version of a {\em locked} tree used in \cite{DDFEHIOOUY2015}.
Using the discussion in \cite{DDFEHIOOUY2015}, 
we obtain the condition for the immovable independent set on a caterpillar:
\begin{theorem}[\cite{DDFEHIOOUY2015}]
\label{th:cat-fix}
Let $G$ and $\bfI$ be a caterpillar and an independent set of $G$, respectively.
Then we cannot slide any token in $\bfI$ on $G$ at all
if and only if there exist a set of locked paths $P_1,\ldots,P_{h}$ for some $h$
such that $\bfI$ is a union of them.
\end{theorem}
The proof can be found in \cite{DDFEHIOOUY2015}, and omitted here.
Intuitively, for any caterpillar $G$ and its independent set $\bfI$,
if $\bfI$ contains a locked path $P$, we cannot slide any token through the vertices in $P$.
Therefore, $P$ splits $G$ into two subgraphs, and we obtain two completely separated subproblems.
(We note that the endpoints of $P$ are leaves with tokens, and their neighbors are spine vertices without
tokens. This property admits us to cut the graph at the spine vertices on the locked path.)
Therefore, we obtain the following lemma:
\begin{lemma}
\label{lem:cat2}
For any given caterpillar $G=(S\cup L,E)$ and two independent sets $\bfI_b$ and $\bfI_r$ on $G$,
there is a linear time reduction from them to another caterpillar 
$G'=(S'\cup L',E')$ and two independent sets $\bfI_b'$ and $\bfI_r'$ such that 
(1) $G$, $\bfI_b$, and $\bfI_r$ are a yes-instance of the {\sc sliding token} problem 
if and only if $G'$, $\bfI_b'$, and $\bfI_r$ are a yes-instance of the {\sc sliding token} problem, and
(2) both of $\bfI_b'$ and $\bfI_r'$ contain no locked path.
\end{lemma}
\begin{proof}
In $G$, when $\bfI_b$ contains a locked path $P$, 
it should be appear in $\bfI_r$; otherwise, the answer is no.
Therefore, we can remove all vertices in $P$ and obtain the new graph 
$G''$ with two independent sets $\bfI_b'':=\bfI_b\setminus P$ and $\bfI_r:=\bfI_r\setminus P$ such that 
$G$ with $\bfI_b$ and $\bfI_r$ is a yes-instance if and only if 
$G''$ with $\bfI_b''$ and $\bfI_r''$ is a yes-instance.
Repeating this process, we obtain disconnected caterpillar $\hat{G}$ and two independent sets $\hat{\bfI_b}$ and 
$\hat{\bfI_r}$ such that both of $\hat{\bfI_b}$ and $\hat{\bfI_r}$ contain no locked paths.
On a disconnected graph, we can solve the problem separately for each connected component.
Therefore, we can assume that the graph is connected, which completes the proof.
\qed
\end{proof}

Hereafter, without loss of generality, 
we assume that the caterpillar $G$ with two independent sets $\bfI_b$ and $\bfI_r$ satisfies the conditions 
in Lemmas \ref{lem:cat} and \ref{lem:cat2}.
That is, each spine vertex $s_i$ has at most one leaf $\ell_i$, 
$s_1$ and $s_{n'}$ have one leaf $\ell_1$ and $\ell_{n'}$, respectively,
both of $\bfI_b$ and $\bfI_r$ contain no locked path, and $\msize{\bfI_b}=\msize{\bfI_r}$.
By the result in \cite{DDFEHIOOUY2015}, this is a yes-instance.
Thus, it is sufficient to show an $O(n^2)$ time algorithm that 
computes a shortest reconfiguration sequence between $\bfI_b$ and $\bfI_r$.

Each pair $(s_i,\ell_i)$ can have at most one token.
Therefore, without loss of generality, we can assume that the blue vertices $b_1, b_2, \ldots, b_k$ in $\bfI_b$
are labeled from left to right (according to the order $(s_1,\ell_1)$, $(s_2,\ell_2)$, $\ldots$, $(s_{n'},\ell_{n'})$ of $G$), 
that is, $L(b_i) < L(b_j)$ if $i < j$; similarly, 
the red vertices $r_1, r_2, \ldots, r_k$ are also labeled from left to right. 
Then, we define a target-assignment $\fallmap: \bfI_b \to \bfI_r$, as follows:
for each blue vertex $b_i \in \bfI_b$
\begin{equation}
\label{eq:map_caterpillar}
 	\fallmap(b_i) = r_i.
\end{equation}
To prove Theorem~\ref{th:caterpillar}, it suffices to show that $\fallmap$ is proper, 
and we can slide tokens with fewest detours.
Here, any token cannot bypass the other token since each token is on a leaf or spine vertex.
Thus, by the results in \cite{DDFEHIOOUY2015}, it has been shown that $\fallmap$ is proper.
We show that we can compute a shortest reconfiguration in case analysis.

Now we introduce {\em direction} of a token $t$ denoted by $dir(t)$ as follows:
when $t$ slides from $v_i\in \{s_i,\ell_i\}$ in $\bfI_b$ to $v_j\in \{s_j,\ell_j\}$ in $\bfI_r$ with $i<j$,
the direction of $t$ is said to be {\em R} and denoted by $dir(t)=R$.
If $i>j$, it is said to be {\em L} and denoted by $dir(t)=L$.
If $i=j$, the direction of $t$ is said to be {\em C} and denoted by $dir(t)=C$.

\begin{figure}[t]
\begin{center}
\includegraphics[width=0.8\textwidth]{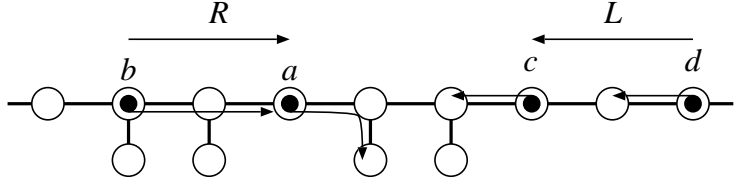}
\end{center}
\caption{The most right R token $a$ has to precede the most left L token $c$.}
\label{fig:rl}
\end{figure}

We first consider a simple case: all directions are either R or L.
In this case, we can use the same idea appearing in the algorithm for 
a proper interval graph in Section \ref{sec:proper}.
We can introduce a partial order over the tokens, 
and slide them straightforwardly using the same idea in 
Section \ref{subsec:algo_proper}.
Intuitively, a sequence of R tokens are slid from left to right,
and a sequence of L tokens are slid from right to left,
and we can define a partial order over the sequences of different directions.
The only additional considerable case is shown in \figurename~\ref{fig:rl}.
That is, when the token $a$ slides to $\ell_i$ from left and
the other token $c$ slides to $s_{i+1}$ from right, $a$ should precede $c$.
It is not difficult to see that this (and its symmetric case)  is the only exception 
than the algorithm in Section \ref{subsec:algo_proper} 
when all tokens slide to right or left.
In other words, in this case, detour is required, and unavoidable.

We next suppose that $\bfI_b$ (and hence $\bfI_r$) contains some token $t$ with $dir(t)=C$. 
In other words, $t$ is put on $s_i$ or $\ell_i$ for some $i$ in both of $\bfI_b$ and $\bfI_r$.
We have five cases.
\begin{listing}{bbb}
\item[Case (1):]
$t$ is put on $\ell_i$ in $\bfI_b$ and $\bfI_r$.
In this case, we have nothing to do; $t$ does not need to be slid.

\item[Case (2):]
$t$ is put on $s_i$ in $\bfI_b$ and slid to $\ell_i$ in  $\bfI_r$.
In this case, we first slide it from $s_i$ to $\ell_i$, and do nothing any more.
Then no detour is needed for $t$.

\item[Case (3):]
$t$ is put on $\ell_i$ in $\bfI_b$ and slid to $s_i$ in  $\bfI_r$.
In this case, we lastly slide it from $\ell_i$ to $s_i$, and 
no detour is needed for $t$ again.

\item[Case (4):]
$t$ is put on $s_i$ in $\bfI_b$ and $\bfI_r$, and $\ell_i$ exists.
Using a simple induction by the number of tokens, we can determine if $t$ should make a detour or not in linear time.
If not, we never slide $t$. Otherwise, we first slide $t$ to $\ell_i$, and lastly slide back from $\ell_i$ to $s_i$.
It is clear that the length of detour with respect to $t$ is as few as possible.


\item[Case (5):]
$t$ is put on $s_i$ in $\bfI_b$ and $\bfI_r$, and $\ell_i$ does not exist.
By assumption, $1<s<n'$ (since $\ell_1$ and $\ell_{n'}$ exist).
Without loss of generality, we suppose $t$ is the leftmost spine vertex having the condition.
We first observe that $\msize{\bfI_b\cap\{s_{i-1},\ell_{i-1},s_{i+1},\ell_{i+1}\}}$ is at most 1.
Clearly, we have no token on $s_{i-1}$ and $s_{i+1}$. 
When we have two tokens on $\ell_{i-1}$ and $\ell_{i+1}$, the path 
$(\ell_{i-1},s_{i-1},s_i,s_{i+1},\ell_{i+1})$ is a locked path, which contradicts the assumption.
We also have $\msize{\bfI_r\cap\{s_{i-1},\ell_{i-1},s_{i+1},\ell_{i+1}\}}\le 1$ by the same argument.

Now we consider the most serious case since the other cases are simpler and easier than this case.
The most serious case is that a blue token on $\ell_{i-1}$ and a red token on $\ell_{i+1}$.
Since any token cannot bypass the other, $\bfI_b$ contains an L token on $\ell_{i-1}$, and 
$\bfI_r$ contains an L token on $\ell_{i+1}$.
In this case, by the L token on $\ell_{i-1}$, first, $t$ should make a detour to right,
and by the L token in $\bfI_r$, $t$ next should make a detour to left twice after the first detour.
It is clear that this three slides should not be avoided, and this ordering of three slides cannot be violated.
Therefore, $t$ itself should slide at least four times to return to the original position, 
and $t$ can done it in four slides.
During this slides, since $t$ is the leftmost spine with this condition,
the tokens on $s_1,\ell_1,s_2,\ell_2,\ldots,s_{i-1},\ell_{i-1}$ do not make any detours.
Thus we focus on the tokens on $s_{i+1},\ell_{i+1},\ldots$.
Let $t'$ be the token that should be on $\ell_{i+1}$ in $\bfI_r$.
Since $t$ is on $s_i$, $t'$ is not on $\{s_{i+1},\ell_{i+1}\}$.
If $t'$ is on one of $\ell_{i+2},s_{i+3},\ell_{i+3},s_{i+4},\ldots$ in $\bfI_b$, we have nothing to do; 
just make a detour for only $t$.
The problem occurs when $t'$ is on $s_{i+2}$ in $\bfI_b$.
If there exists $\ell_{i+2}$, we first slide $t'$ to it, and this detour for $t'$ is unavoidable.
If $\ell_{i+2}$ does not exist, we have to slide $t'$ to $s_{i+3}$ before slide of $t$.
This can be done immediately except the only considerable case; 
when we have another L or S token $t''$ on $s_{i+3}$.
We can repeat this process recursively and confirm that each detour is unavoidable.
Since $G$ with $\bfI_b$ and $\bfI_r$ contains no locked path, this process will halts.
(More precisely, this process will be stuck if and only if this sequence of tokens forms a locked path on $G$,
which contradicts the assumption.)
Therefore, traversing this process, we can construct the shortest reconfiguration sequence.
\end{listing}

\noindent
{\bf Proof of Theorem~\ref{th:caterpillar}.}
For a given independent set $\bfI_b$ on a caterpillar $G=(V,E)$,
we can check if each vertex is a part of locked path as follows in $O(n)$ time 
(which is much simpler than the algorithm in \cite{DDFEHIOOUY2015}):
\begin{listing}{aaa}
\item[(0)] Initialize a state $S$ by ``not locked path''.
\item[(1)] For $i=1,2,\ldots,n'$, check $s_i$ and $\ell_i$.
   We here denote their states by $(s_i,\ell_i)=(x,y)$, 
   where $x\in \{0,1\}$, $y\in \{0,1,-\}$ such that $1$ means ``token is placed on the vertex'',
   $0$ means ``no token is placed on the vertex'', and $-$ means ``the leaf does not exist.''
   In each case, update the state $S$ as follows:
   \begin{listing}{bbb}
    \item[Case $(0,1)$:] If $S$ is ``not locked path,'' set $S$ by ``locked path?,'' 
  	       and remember $i$ as a potential left endpoint of a locked path.
               If $S$ is ``locked path'', $(s_i,\ell_i)$ is a part of locked path.
	       Therefore, mark all vertices between the previously remembered 
	       left endpoint to this endpoint as ``locked path''.
	       After that, set $S$ by ``locked path?'' again, and remember $i$ as 
	       a potential left endpoint of the next locked path.
    \item[Case $(0,-)$ and $(0,0)$:] If there is a token on $s_{i-1}$, nothing to do. 
	       If there is no token on $s_{i-1}$, 
	       reset $S$ by ``not locked path'' (regardless of the previous state of $S$).
    \item[Case $(1,0)$:] reset $S$ by ``not locked path'' (regardless of the previous state of $S$).
    \item[Case $(1,-)$:] nothing to do. 
 \end{listing}
\end{listing}
Simple case analysis shows that after the procedure above, every vertex in 
a locked path is marked in $O(n)$ time.
Thus, we first run this procedure twice for $(G,\bfI_b)$ and $(G,\bfI_r)$ in $O(n)$ time, 
and check whether the marked vertices coincide with each other. 
If not, the algorithm outputs ``no''.  
Otherwise, the algorithm splits the caterpillar $G$ into 
subgraphs $G_1,G_2,\ldots,G_h$ induced by only unmarked vertices.
Then we can solve the problem for each subgraph; we note that 
two endpoints of which tokens are placed of a locked path $P$ are leaves.
That is, for example, when a locked path $P=(p_1,p_2,\ldots,p_k)$ splits $G$ into $G_1$ and $G_2$,
the neighbors of $G_1$ and $G_2$ in $P$ are $p_2$ and $p_{k-1}$, and there are no token on them.
Thus, in the case, we can solve the problem on $G_1$ and $G_2$ separately, and 
we do not need to consider their neighbors.

For each subgraph $G_1,\ldots,G_h$, the algorithm next checks whether 
each subgraph contains the same number of blue and red tokens.
If they do not coincide with each other, the algorithm output ``no.''
Otherwise, we have a yes-instance. 
The correctness of the algorithm so far follows from Theorem \ref{th:cat-fix} 
with results in \cite{DDFEHIOOUY2015} immediately.
It is also easy to implement the algorithm to run in $O(n)$ time and space.

It is not difficult to modify the algorithm to output the sequence itself based on the previous 
case analysis. For each token, the number of detours made by the token is bounded above by $O(n)$, 
the number of slides of the token itself is also bounded above by $O(n)$, 
and the computation for the token can be done in $O(n)$ time.
Therefore, the algorithm runs in $O(n^2)$ time, and the length of the sequence is $O(n^2)$.
(As shownn in the last paragraph in Section \ref{sec:proper},
there exist instances that require a shortest sequence of length $\Theta(n^2)$.)
\qed

\section{Concluding Remarks}

In this paper, we showed that the {\sc shortest sliding token} problem can be solved 
in polynomial time for three subclasses of interval graphs.
The computational complexity of the problem for chordal graphs, interval graphs, and trees are still open.
Especially, tree seems to be the next target.
We can decide if two independent sets are reconfigurable in linear time \cite{DDFEHIOOUY2015}, 
then can we find a shortest sequence for a yes-instance in polynomial time?
As in the 15-puzzle, finding a shortest sequence can be NP-hard.
For a tree, we do not know that the length can be bounded by any polynomial or not.
It is an interesting open question whether there is 
any instance on some graph classes whose reconfiguration sequence requires super-polynomial length.

\bibliographystyle{abbrv}

\begin{thebibliography}{10}

\bibitem{BogartWest1999}
Bogart, K.P., West, D.B.:
A short proof that `proper=unit'.
Discrete Mathematics 201, pp.~21--23 (1999)

\bibitem{BJLPP11}
Bonamy, M., Johnson, M., Lignos, I., Patel, V., Paulusma, D.: 
On the diameter of reconfiguration graphs for vertex colourings. 
Electronic Notes in Discrete Mathematics 38, pp.~161--166 (2011)

\bibitem{BC09}
Bonsma, P., Cereceda, L.: 
Finding paths between graph colourings: PSPACE-completeness and superpolynomial distances. 
Theoretical Computer Science 410, pp.~5215--5226 (2009)  

\bibitem{BonsmaKaminskiWrochna}
Bonsma, P., Kami\'nski, M., Wrochna M.:
Reconfiguration Independent Sets in Claw-Free Graphs
arXiv:1403.0359, 2014.



\bibitem{BLS99} 
Brandst\"adg, A., Le, V.B., Spinrad, J.P.:
Graph Classes: A Survey, SIAM (1999)

\bibitem{CHJ11}
Cereceda, L., van den Heuvel, J., Johnson, M.:  
Finding paths between 3-colourings. 
J.~Graph Theory 67, pp.~69--82 (2011)

\bibitem{DDFEHIOOUY2015}
Demaine, E.D., Demaine, M.L., Fox-Epstein, E., Hoang, D.A., Ito T.,
Ono, H., Otachi, Y., Uehara, R., Yamada, T.:
Linear-Time Algorithm for Sliding Tokens on Trees.
Theoretical Computer Science 600, pp.~132--142 (2015)  


\bibitem{DengHellHuang1996}
Deng, X., Hell, P., Huang., J.:
Linear-time representation algorithms for proper circular-arc graphs and proper interval graphs.
SIAM J.~Computing 25, pp.~390--403 (1996)

\bibitem{FoxEpsteinHoangOtachiUehara2015}
Fox-Epstein, E., Hoang, D.A., Otachi, Y., Uehara, R.:
Sliding Token on Bipartite Permutation Graphs.
The 26th International Symposium on Algorithms and Computation (ISAAC), accepted, 2015.

\bibitem{Gardner}
Gardner, M.:
The Hypnotic Fascination of Sliding-Block Puzzles.
Scientific American 210, pp. 122--130 (1964).

\bibitem{Kolaitis}
Gopalan, P., Kolaitis, P.G., Maneva, E.N., Papadimitriou, C.H.:
The connectivity of Boolean satisfiability: computational and structural dichotomies.
SIAM J.~Computing 38, pp.~2330--2355 (2009)  

\bibitem{HearnDemaine2005}
Hearn, R.A., Demaine, E.D.: 
PSPACE-completeness of sliding-block puzzles and other problems through the nondeterministic constraint logic model of computation. 
Theoretical Computer Science 343, pp.~72--96 (2005) 

\bibitem{HearnDemaine2009}
Hearn, R.A., Demaine, E.D.: 
Games, Puzzles, and Computation.
A K Peters (2009)


\bibitem{IDHPSUU}
Ito, T., Demaine, E.D., Harvey, N.J.A., Papadimitriou, C.H., Sideri, M., Uehara, R., Uno, Y.: 
On the complexity of reconfiguration problems.
Theoretical Computer Science 412, pp.~1054--1065 (2011)



 

\bibitem{KMP11}
Kami\'nski, M., Medvedev, P., Milani${\rm \check{c}}$, M.: 
Shortest paths between shortest paths.
Theoretical Computer Science 412, pp.~5205--5210 (2011)

\bibitem{KaminskiMedvedevMilanic2012}
Kami\'nski, M., Medvedev, P., Milani${\rm \check{c}}$, M.: 
Complexity of independent set reconfigurability problems.
Theoretical Computer Science 439, pp.~9--15 (2012)

\bibitem{KM89}
Korte, N., M{\"{o}}hring, R.:
An incremental linear-time algorithm for recognizing interval graphs.
SIAM J.~Computing 18, pp.~68--81 (1989)



         
\bibitem{MTY11}
Makino, K., Tamaki, S., Yamamoto, M.:
An exact algorithm for the Boolean connectivity problem for $k$-CNF.
Theoretical Computer Science 412, pp.~4613--4618 (2011)

\bibitem{MNPR15}
Mouawad, A.E., Nishimura, N., Pathak, V., Raman, V.: 
Shortest Reconfiguration Paths in the Solution Space of Boolean Formulas.
In Proc.~of ICALP 2015, LNCS 9134, pp. 985--996 (2015) 

\bibitem{MNRSS13}
Mouawad, A.E., Nishimura, N., Raman, V., Simjour, N., Suzuki, A.: 
On the parameterized complexity of reconfiguration problems.
In Proc.~of IPEC 2013, LNCS 8296, pp. 281--294 (2013) 

\bibitem{MouawadNishimuraRamanWrochna}
Mouawad, A.E., Nishimura, N., Raman, V., Wrochna, M.: 
Reconfiguration over tree decompositions.
In Proc.~of IPEC 2014, LNCS 8894, pp. 246--257 (2014)


\bibitem{RW90}
Ratner, R., Warmuth, M.:
Finding a shortest solution for the $N\times N$-extension of the 15-puzzle is intractable.
J. Symb. Comp., Vol. 10, pp. 111--137, 1990.

\bibitem{Slocum}
Slocum, J.: 
The 15 Puzzle Book: How it Drove the World Crazy.
Slocum Puzzle Foundation, 2006.

\end{thebibliography}








\end{document}